\RequirePackage{amsthm}
\documentclass[sn-mathphys-num,pdflatex]{sn-jnl}% Math and Physical Sciences Reference Style

\usepackage{graphicx}%
\usepackage{multirow}%
\usepackage{amsmath,amssymb,amsfonts}%
\usepackage{amsthm}%
\usepackage{mathrsfs}%
\usepackage[title]{appendix}%
\usepackage{xcolor}%
\usepackage{textcomp}%
\usepackage{manyfoot}%
\usepackage{booktabs}%
\usepackage{listings}%

\usepackage{tabularray}
\usepackage{overpic}
\usepackage[normalem]{ulem}

\theoremstyle{thmstyleone}%
\newtheorem{theorem}{Theorem}%  meant for continuous numbers
\newtheorem{corollary}[theorem]{Corollary}% 
\newtheorem{lemma}[theorem]{Lemma}%

\theoremstyle{thmstyletwo}%
\newtheorem{example}{Example}%
\newtheorem{remark}{Remark}%

\theoremstyle{thmstylethree}%
\newtheorem{definition}{Definition}%

\usepackage{mathtools}
\DeclarePairedDelimiter\abs{\lvert}{\rvert}

\DeclarePairedDelimiter\ceil{\lceil}{\rceil}
\DeclarePairedDelimiter\floor{\lfloor}{\rfloor}
\DeclarePairedDelimiter\parenv{\lparen}{\rparen}
\DeclarePairedDelimiter\sparenv{\lbrack}{\rbrack}

\DeclarePairedDelimiter\set{\{}{\}}

\renewcommand{\leq}{\leqslant}

\renewcommand{\geq}{\geqslant}

\newcommand{\cS}{\mathcal{S}}

\newcommand{\F}{\mathbb{F}}
\newcommand{\C}{\mathbb{C}}

\newcommand{\Z}{\mathbb{Z}}

\newcommand{\Q}{\mathbb{Q}}

\newcommand{\eqdef}{\triangleq}

\DeclareMathOperator{\supp}{supp}
\DeclareMathOperator{\wt}{wt}

\begin{document}

\title[On Nearly-Perfect Covering Codes Beyond Radius One]{On Nearly-Perfect Covering Codes Beyond Radius One}


\author*[1]{\fnm{Gabriel} \sur{Sac Himelfarb}}\email{sachimeg@mcmaster.ca}

\author[1,2]{\fnm{Moshe} \sur{Schwartz}}\email{schwartz.moshe@mcmaster.ca}

\affil[1]{\orgdiv{Department of Electrical and Computer Engineering}, \orgname{McMaster University}, \orgaddress{\city{Hamilton}, \postcode{L8S 4K1}, \state{ON}, \country{Canada}}}

\affil[2]{\orgdiv{School of Electrical and Computer Engineering}, \orgname{Ben-Gurion University of the Negev}, \orgaddress{\city{Beer Sheva}, \postcode{8410501}, \country{Israel}}}


\abstract{
We study (binary) nearly-perfect covering codes, which are codes that attain the Van Wee bound with equality. They act as the covering counterparts to nearly-perfect error-correcting codes, which attain the Johnson bound with equality. These codes have been completely classified for covering radius $R=1$. We prove that no code with $R\geq 2$ can attain the original Van Wee bound with equality, since it omits the dependence on the minimum distance of the code. We refine the bound to account for the minimum distance and show some nearly-perfect covering codes. By proving some structural properties of such codes, we prove all nearly-perfect covering codes with $R=2,3$ must be equivalent to the codes we showed. We also prove that for any $R\geq 3$, there are at most a finite number of nearly-perfect covering codes.
}




\maketitle
\section{Introduction}

A (binary) code $C$ of length $n$ and size $M$ is a subset $C\subseteq \F_2^n$ of size $\abs{C}=M$. The elements of $C$ are called codewords. We equip $\F_2^n$ with the Hamming metric defined by the distance function between $x=(x_1,\dots,x_n),y=(y_1,\dots,y_n)\in\F_2^n$ as
\[
d(x,y) \eqdef \abs*{\set*{ 1\leq i\leq n : x_i\neq y_i}}.
\]
With this metric in place, two important parameters of $C$ are its minimum distance $d$, and covering radius $R$, defined as
\begin{align*}
d &\eqdef \min_{\substack{c,c'\in C \\ c\neq c'}} d(c,c'), \\
R &\eqdef \max_{v\in\F_2^n} \min_{c\in C} d(v,c).
\end{align*}
We summarize all the code parameters by saying $C$ is an $(n,M,d)R$ code.

A geometric interpretation of a code requires the definition of a ball of radius $r$, centered at $v\in\F_2^n$ as follows,
\[
B_r(v) \eqdef \set*{v'\in \F_2^n : d(v,v')\leq r}.
\]
The size of a ball does not depend on its center, $v$, and equals
\[
\abs*{B_r(v)} = \sum_{i=0}^r \binom{n}{i}.
\]
The packing radius of $C$ is the largest integer $e$ such that balls of radius $e$ centered at the codewords of $C$ are pairwise-disjoint. That is, for all $c,c'\in C$, $c\neq c'$, we have $B_e(c)\cap B_e(c') = \emptyset$. It is readily seen that $e=\floor{(d-1)/2}$, and we have the ball-packing bound
\begin{equation}
\label{eq:bpb}
M\sum_{i=0}^{e} \binom{n}{i} \leq 2^n. 
\end{equation}
Similarly, the covering radius of $C$ is the smallest integer $R$ such that balls of radius $R$ centered at the codewords of $C$ cover the entire space, i.e., $\bigcup_{c\in C} B_R(c) = \F_2^n$. This matches the definition of $R$ given above, and we have the ball-covering bound
\begin{equation}
\label{eq:bcb}
M\sum_{i=0}^{R} \binom{n}{i} \geq 2^n. 
\end{equation}

In the extreme case where the packing and covering radii of $C$ are the same, then the balls of radius $R=\floor{(d-1)/2}$ centered at the codewords of $C$ tile $\F_2^n$, and the code is said to be \emph{perfect}. For such codes, both~\eqref{eq:bpb} and~\eqref{eq:bcb} are attained with equality. Apart from the trivial perfect codes ($C=\F_2^n$, or any $C$ with $\abs{C}=1$), 
other perfect codes have been known for quite some time in the following parameters: $(2m+1,2,2m+1)m$ (e.g., the binary repetition code of odd length), $(2^m-1,2^{2^m-1-m},3)1$ (e.g., the binary Hamming code), and $(23,4096,7)3$ (e.g., the binary Golay code). We note that for each set of parameters there may exist many other equivalent and non-equivalent codes with the same parameters. It was the work of Tiet{\"a}v{\"a}inen and Perko~\cite{TiePer71} that completed the classification by showing that no other possible parameters exist for perfect codes.

A natural question to ask is: If a code is not perfect, can we improve on the ball-packing bound of~\eqref{eq:bpb}? The answer is yes. The Johnson bound (see~\cite{Joh62} and~\cite[Theorem 2.3.8]{PleHuf03}) states that for any $(n,M,2e+1)R$ code,
\begin{equation}
\label{eq:jbodd}
M\parenv*{\sum_{i=0}^e \binom{n}{i}+\frac{\binom{n}{e}}{\ceil{\frac{n-e}{e+1}}}\parenv*{\frac{n+1}{e+1}-\floor*{\frac{n+1}{e+1}}}}\leq 2^n,
\end{equation}
and for any $(n,M,2e+2)R$ code,
\begin{equation}
\label{eq:jbeven}
M\parenv*{\sum_{i=0}^e \binom{n}{i}+\frac{\binom{n}{e+1}}{\ceil{\frac{n-e}{e+1}}}}\leq 2^n.
\end{equation}

A non-perfect code that meets~\eqref{eq:jbodd}-\eqref{eq:jbeven} with equality is called a nearly-perfect error-correcting code. Once again, nearly-perfect error-correcting codes have been known in the following parameters: $(2^m-2, 2^{2^m-2-m},3)2$ (e.g., shortened binary Hamming codes), and $(2^m-1,2^{2^m-2m},5)3$ (e.g., punctured Preparata codes). Lindstr{\"o}m~\cite{Lin75} completed the classification of nearly-perfect error-correcting codes by showing no other parameters are possible (see also~\cite[Section 2.3.5]{PleHuf03} and the references therein).

In an analogous fashion, when there is no perfect code, the ball-covering bound of~\eqref{eq:bcb} can be improved. This was first proved by Van Wee in~\cite{Van88}, and later the proof was simplified by Struik in~\cite{Str94}. The improved bound states that for every $(n,M,d)R$ code,
\begin{equation}
\label{eq:vwb}
M \parenv*{\sum_{i=0}^R\binom{n}{i}-\frac{\binom{n}{R}}{\ceil{\frac{n-R}{R+1}}}\parenv*{\ceil*{\frac{n+1}{R+1}}-\frac{n+1}{R+1}}}\geq 2^n.
\end{equation}
The striking similarity between the Johnson bound~\eqref{eq:jbodd} and the Van Wee bound~\eqref{eq:vwb}, begs the question of finding the possible parameters for nearly-perfect covering codes (non-perfect codes that attain~\eqref{eq:vwb} with equality). Surprisingly, this has not been attempted until recently. Boruchovsky \emph{et al.}~\cite{BorEtzRot25} started studying this question by completely classifying all nearly-perfect covering codes with $R=1$. Among other things, it is proved in~\cite{BorEtzRot25} that the parameters of nearly-perfect covering codes with $R=1$ must either be $(2^m, 2^{2^m-m},2)1$ or $(2^m,2^{2^m-m},1)1$, and nearly-perfect covering codes exist in all of them.

We list our main contributions: We first examine the Van Wee bound, and show that as it was stated, no nearly-perfect covering code with covering radius $R\geq 2$ can exist. We note that this is due to the fact that it omits the dependence on the minimum distance of the code. After refining the bound to account for the minimum distance we can proceed to define nearly-perfect codes with covering radii $R\geq 2$. We show four codes that are nearly-perfect covering, two of which are parametric. We prove some structural properties of nearly-perfect covering codes, showing they must be quasi-perfect, as well as study the way over-covered vectors are located with respect to codewords. With this knowledge, we employ number-theoretic results to classify all possible parameters for the case of $R=2$, as well as show that all nearly-perfect codes of that covering radius must be equivalent to the codes we showed before. We extend the analysis and classify all nearly-perfect covering codes for $R=3$, and show how this approach generalizes to any higher radius. Additionally, we prove that for any covering radius $R\geq 3$ there are at most a finite number of nearly-perfect covering codes.

This paper is organized as follows. In Section~\ref{sec:vanwee} we show the relevant parts of the proof of the original Van Wee bound, and refine it to account for the minimal distance of the code, as well as show some nearly-perfect covering codes. In Section~\ref{sec:properties} we study properties of nearly-perfect covering codes. We proceed in Section~\ref{sec:R2} to completely classify all nearly-perfect covering codes with covering radius $R=2$, and then extend the analysis to higher radii in Section~\ref{sec:R3ormore}. In Section~\ref{sec:groupalgebra} we generalize regularity results from~\cite{BorEtzRot25}. We conclude in Section~\ref{sec:conc} with a short discussion of the results and open questions.

\section{Refining the Van Wee Bound}
\label{sec:vanwee}

We describe the original Van Wee bound~\cite{Van88} on the parameters of covering codes. We give some details on parts of its proof for two reasons. First, it allows us to point out a straightforward refinement which also takes into account the minimum distance of the code, which was neglected in the original theorem. Second, it allows us later to study the inner structure of codes attaining the bound with equality. In our description we also follow the notation and proof given by~\cite{Str94}.

We require some definitions. Throughout the paper, we let $C$ denote an $(n,M,d)R$ code. The support of a vector $x=(x_1,\dots,x_n)\in\F_2^n$ is $\supp(x)\eqdef \set{1\leq i\leq n : x_i\neq 0}$, and the weight $\wt(x)$ is the size of its support, or equivalently, $\wt(x)\eqdef \abs{\supp(x)} = d(x,0)$. The support of a set of codewords, $X\subseteq C$, is given by $\supp(X)\eqdef \bigcup_{x\in X}\supp(x)$. We denote by $e_i$ the standard $i$-th unit vector, i.e., the vector all of whose entries are $0$ except for the $i$-th one which contains $1$.

We use the notation $d(x,C)$ to denote the distance from a vector $x\in\F_2^n$ to the code $C$, i.e., $d(x,C)\eqdef\min_{c\in C}d(x,c)$. Thus, $R=\max_{x\in\F_2^n} d(x,C)$. We use $C_i$ to denote the set of vectors at distance exactly $i$ from $C$, namely, $C_i\eqdef \set{ x\in\F_2^n : d(x,C)=i}$. In particular, $C_R$ is the set of points farthest away from the code, which are also known as the \emph{deep holes} of $C$.

Given $x\in \F_2^n$, we say it is covered by a codeword $c\in C$ if $x\in B_R(c)$, or equivalently, $c\in B_R(x)$. We say $x$ is \emph{over-covered} by the code $C$ if $\abs{B_R(x)\cap C} \geq 2$, namely, $x$ is contained in at least two distinct balls of radius $R$ around the codewords of $C$. Let $f(x)$ denote the number of codewords at distance at most $R$ from $x$, and let $g(x)$ be the over-covering of $x$, i.e.,
\begin{align*}
f(x) &\eqdef \abs*{B_R(x) \cap C}, & g(x) &\eqdef f(x)-1.
\end{align*}
We extend the definition of $f$ and $g$ to sets of vectors, $X\subseteq \F_2^n$, as 
\begin{align*}
f(X)&\eqdef \sum_{x\in X}f(x), & g(X)&\eqdef \sum_{x\in X}g(x)=f(X)-\abs{X}.
\end{align*}
Let $\varepsilon$ be the average over-covering of all balls of radius $1$ centered at elements of $C_R$,
\[
\varepsilon \eqdef \frac{1}{\abs{C_R}}\sum_{x\in C_R}g(B_1(x)).
\]
Let $\mu$ be the maximum value of $\abs{B_1(x)\cap C_R}$ among all over-covered $x\in \F_2^n$, i.e. all $x$ such that $g(x)>0$, i.e.,
\begin{equation}
\label{eq:defmu}
\mu \eqdef \max_{\substack{x\in\F_2^n \\ g(x)>0}} \abs*{B_1(x)\cap C_R}.
\end{equation}

\begin{lemma}[{\cite[Eq.~(6)]{Str94}}]\label{lemma:epsilon}
    For any $x\in C_R$, 
    \begin{equation}\label{eq:ineqepsilon}
        g(B_1(x)) \geq (R+1)\cdot \parenv*{\ceil*{\frac{n+1}{R+1} }-\frac{n+1}{R+1}},
    \end{equation}
    with equality if and only if $\abs{B_{R+1}(x)\cap C}=\ceil{\frac{n+1}{R+1}}$. 
\end{lemma}

For the proof we refer the reader to \cite{Str94}, since we will not need the details for our subsequent study. The next lemma is also from~\cite{Str94}, and we provide its proof since its inner workings are used by us later.

\begin{lemma}[{\cite[Eq.~(8)]{Str94}}]\label{lemma:muineq}
If $z\in \F_2^n$ is such that $g(z)>0$, then 
\[
\abs*{B_1(z)\cap C_R}\leq \mu \leq n+1-R-\ceil{d/2}.
\]
\end{lemma}

\begin{proof}
We first note that it suffices to prove that for all $z$ as above, $\abs{B_1(z)\cap C_R}\leq n+1-R-\ceil{d/2}$, since then~\eqref{eq:defmu} completes the claim. Define
\[
\widetilde{T}(z)\eqdef\set*{y\in\F_2^n: z-y\in C \text{ and } \wt(y)\leq R}.
\]
We know that $z$ is covered by (at least) two codewords $c,c'$ at distance at most $R$ from $z$. Then $y\eqdef z-c$ and $y'\eqdef z-c'$ both belong to $\widetilde{T}(z)$. We denote $A\eqdef \supp(y)$ and $A'\eqdef \supp(y')$.

Notice that if $d(z,C)<R-1$, then $\abs{B_1(z)\cap C_R}=0$, so the inequality holds\footnote{It is possible to see that $R\leq n-\ceil{d/2}$: Let $c,c'\in C$ be two codewords such that $d(c,c')=d$. Without loss of generality, assume $c$ and $c'$ differ in the first $d$ coordinates. Given any $x\in\F_2^n$, we can flip at most $\floor{d/2}$ bits so the first $d$ coordinates match $c$ or $c'$, and then flip at most $n-d$ bits to match the last $n-d$ coordinates to both $c$ and $c'$. Thus, $d(x,C)\leq \floor{d/2}+n-d = n-\ceil{d/2}$.}, and in the remainder of the proof we may suppose that $d(z,C)$ is either $R-1$ or $R$. We study the two cases separately. 

\textbf{Case 1:} $d(z,C)=R$, so $d(z,c)=d(z,c')=R$. Let us see that $\abs{\supp (\widetilde{T}(z))}\geq R+\lceil d/2\rceil$. Since both $y$ and $y'$ belong to $\widetilde{T}(z)$, we have
\begin{align*}
\abs*{\supp (\widetilde{T}(z))} &\geq \abs*{A \cup A'}=\abs*{A}+\abs*{A'}-\abs*{A\cap A'} =2R-\abs*{A \cap A'}.
\end{align*}
Notice that
\begin{align*}
d\leq d(c,c')=d(y,y')=\abs*{A}+\abs*{A'}-2\abs*{A\cap A'}
=2R-2\abs*{A\cap A'}.
\end{align*}
Thus,
\[\abs*{A\cap A'}\leq R - \ceil{d/2},\]
and then
\[\abs*{\supp (\widetilde{T}(z))}\geq 2R - (R-\ceil{d/2})= R+ \ceil{d/2}.\]

\textbf{Case 2:} $d(z,C)=R-1$. We can see in a very similar way that $\abs{\supp (\widetilde{T}(z))}\geq R-1+\ceil{d/2}$, but now we need to consider separately the cases when both $d(c,z)=d(c',z)=R-1$ (Case 2.a) and when $d(c,z)=R$, $d(c',z)=R-1$ (Case 2.b):

\textbf{Case 2.a:} If $d(c,z)=d(c',z)=R-1$, then 
\[
\abs*{\supp (\widetilde{T}(z))}\geq \abs*{A\cup A'}=R-1+R-1-\abs*{A\cap A'}=2R-2-\abs*{A\cap A'}.
\]
We observe that
\[
d\leq d(c,c')=\abs{A}+\abs{A'}-2\abs*{A\cap A'}=2R-2-2\abs*{A\cap A'},
\]
which implies $\abs{A\cap A'}\leq R-1-\ceil{d/2}$, so then
\[
\abs*{\supp (\widetilde{T}(z))}\geq R-1+\ceil{d/2}.
\]

\textbf{Case 2.b:} If $d(c,z)=R$ and $d(c',z)=R-1$, then 
\[
\abs*{\supp (\widetilde{T}(z))}\geq \abs*{A\cup A'}=R-1+R-\abs*{A\cap A'}=2R-1-\abs*{A\cap A'},
\]
and 
\[
d\leq d(c,c')=\abs{A}+\abs{A'}-2\abs*{A\cap A'}=2R-1-2\abs*{A\cap A'}.
\]
Thus we have $\abs{A\cap A'}\leq R-\ceil{(1+d)/2}$, and therefore
\[
\abs*{\supp (\widetilde{T}(z))}\geq R-1+\ceil{(d+1)/2} \geq R-1+\ceil{d/2},
\]
where we note that the last inequality is strict when $d$ is even.
    
To conclude the proof, we observe that for $i\in\supp (\widetilde{T}(z))$ we have $d(z+e_i,C)<R$, so $z+e_i\notin C_R$. Thus, if $d(z,C)=R$, then $z\in C_R$, and by Case 1,
\[
\abs*{B_1(z)\cap C_R}\leq 1+n-\abs*{\supp (\widetilde{T}(z))}\leq n+1-(R+\lceil d/2\rceil).
\]
On the other hand, if $d(z,C)=R-1$, then $z\notin C_R$, and by Case 2,
\[
\abs*{B_1(z)\cap C_R}\leq n-\abs*{\supp (\widetilde{T}(z))} \leq n-(R-1+\lceil d/2\rceil).
\]
In both cases we see that the desired inequality follows.
\end{proof}

With these two lemmas it is now possible to prove the Van Wee bound. We rephrase the proof to highlight the point of omitting the dependence on the minimum distance of the code in the original bound.

\begin{theorem}[\cite{Van88,Str94}]
\label{th:vanwee}
For any $(n,M,d)R$ code, \eqref{eq:vwb} holds.
\end{theorem}

\begin{proof}
The proof relies on computing the sum $\sum_{x\in C_R}g(B_1(x))$  in two different ways. We have the following chain of equalities and inequalities:
\begin{align}
    (n-R)\parenv*{M\sum_{i=0}^R \binom{n}{i}-2^n} 
    & \geq (n+1-R-\ceil{d/2})\parenv*{M\sum_{i=0}^R \binom{n}{i}-2^n} \label{ineq1}\\
    &\geq \mu \cdot\parenv*{M\sum_{i=0}^R \binom{n}{i}-2^n}\label{ineq2}\\
    &=  \mu \cdot g(\F_2^n) \label{ineq3}\\
    & \geq \sum_{z\in \F_2^n}g(z)\cdot\abs*{B_1(z)\cap C_R} \label{ineq4}\\
    & = \sum_{x\in C_R} \sum_{\substack{z\in\F_2^n \\ d(z,x)\leq 1}}g(z)= \sum_{x\in C_R}g(B_1(x)) = \varepsilon \cdot\abs*{C_R} \nonumber\\
    & \geq \varepsilon \cdot \parenv*{2^n-M\sum_{i=0}^{R-1}\binom{n}{i}}\label{ineq6}\\
    &\geq (R+1)\parenv*{\ceil*{\frac{n+1}{R+1}}-\frac{n+1}{R+1}}\parenv*{2^n-M\sum_{i=0}^{R-1}\binom{n}{i}} \label{ineq7},
\end{align}
where~\eqref{ineq1} follows from $\ceil{d/2}\geq 1$, \eqref{ineq2} follows from Lemma~\ref{lemma:muineq}, \eqref{ineq3} holds because the total covering is equal to the sum of the sizes of all balls of radius $R$ centered at codewords, \eqref{ineq4} follows from the definition of $\mu$, \eqref{ineq6} holds because points at distance $R$ from $C$ cannot be inside any ball of radius $R-1$ centered at a codeword, and~\eqref{ineq7} follows from Lemma~\ref{lemma:epsilon}. By comparing the endpoints of the chain of inequalities, and after some algebraic manipulation, we arrive at the desired bound. 
\end{proof}

In~\cite{BorEtzRot25}, the authors analyze and construct codes which attain the Van Wee bound with $R=1$. The case of nearly-perfect covering codes with $R\geq 2$ has not been studied to the best of our knowledge. Our first result claims that no code with covering radius $R\geq 2$ attains the Van Wee bound with equality.

\begin{theorem}
Let $C$ be an non-perfect $(n,M,d)R$ code, with $R\geq 2$. Then~\eqref{eq:vwb} is not attained with equality.
\end{theorem}
\begin{proof}
Assume to the contrary a code attaining~\eqref{eq:vwb} exists. Then all the inequalities in the proof of Theorem~\ref{th:vanwee} are in fact equalities. Thus, from~\eqref{ineq1}, $\ceil{d/2}=1$, so $d\leq 2$. Additionally, from~\eqref{ineq2} and~\eqref{ineq4}, every $z\in\F_2^n$ which is over-covered (namely, $g(z)>0$) needs to satisfy
\[
\abs{B_1(z)\cap C_R}=\mu=n-R>0,
\]
where the last inequality holds because $M>1$ (since the code is not perfect). Then $d(z,C)\in \set{R-1,R}$ for any over-covered $z$. Take $c,c'\in C$ such that $d(c,c')=d\leq 2$. If we had $R\geq 2$, then $c$ and $c'$ would cover themselves and each other, and thus they would be over-covered. But $d(c,C)=0\notin\set{R-1,R}$, which is a contradiction if $R\geq 2$. 
\end{proof}

This analysis implies that the Van Wee bound is limited in some sense to treat nearly-perfect codes of radius $R=1$ only. If we are looking for a reason for this limitation, it becomes clear that~\eqref{ineq1} is the culprit, since it omits information given by the code's minimum distance $d$. We suggest the following straightforward refinement of the Van Wee bound.

\begin{theorem}
\label{th:refined}
For all $(n,M,d)R$ codes,
\begin{equation}\label{eq:newdef}
        M \parenv*{\sum_{i=0}^R \binom{n}{i}-\binom{n}{R}\frac{\ceil{\frac{n+1}{R+1}}-\frac{n+1}{R+1}}{\ceil{\frac{n+1}{R+1}}-\frac{R+\ceil{d/2}}{R+1}} }\geq 2^n.
\end{equation}
\end{theorem}
\begin{proof}
The proof proceeds along the same sequence of inequalities as the proof of Theorem~\ref{th:vanwee}. However, we omit~\eqref{ineq1}, and compare the two sides,
\begin{multline*}
(n+1-R-\ceil{d/2})\parenv*{M\sum_{i=0}^R \binom{n}{i}-2^n} \\
\geq 
(R+1)\parenv*{\ceil*{\frac{n+1}{R+1}}-\frac{n+1}{R+1}}\parenv*{2^n-M\sum_{i=0}^{R-1}\binom{n}{i}}.
\end{multline*}
After rearranging, we obtain the desired inequality.
\end{proof}

By abuse of terminology, we redefine nearly-perfect codes to match the refined version of the Van Wee bound:

\begin{definition}
We say a non-perfect $(n,M,d)R$ code is a \emph{nearly-perfect covering code} if~\eqref{eq:newdef} holds with equality.
\end{definition}

We note that when $\ceil{d/2}=1$, \eqref{eq:newdef} agrees with~\eqref{eq:vwb}. Thus, our new definition of nearly-perfect codes matches that used in~\cite{BorEtzRot25}, since all the codes there had $\ceil{d/2}=1$.

With this new definition, there do exist nearly-perfect codes with $R\geq 2$, as the following examples show.

\begin{example}\label{ex:repeat}
Let $n$ be even, $n=2m$. We consider the following two codes:
\begin{align*}
    C &= \set*{ 0^{2m},1^{2m}}, & C'&=\set*{0^{2m},1^{2m-1}0}.
\end{align*}
The code $C$ is the well-known binary repetition code, with parameters $(2m,2,2m)m$. The second code is a nearly-repetition code, with parameters $(2m,2,2m-1)m$. We comment that it is straightforward to verify that the covering radius of both codes is indeed $m$: all vectors of weight at most $m$ are at distance at most $m$ from $0^{2m}$, while all the remaining vectors are at distance at most $m$ from $1^{2m}$ and $1^{2m-1}0$, respectively. The two codes are linear, and are not equivalent to each other. We also observe that if we denote the minimum distance of $C$ and $C'$ by $d$ and $d'$, respectively, then both codes have the same packing radius, $\floor{(d-1)/2}=\floor{(d'-1)/2}=m-1$.

We contend that both codes are nearly-perfect covering codes, i.e., they attain~\eqref{eq:newdef} with equality. Since $\ceil{d/2}=\ceil{d'/2}=m$, it is enough that we show this for $C$. Then 
\begin{align*}
\frac{n+1}{R+1}&=\frac{n+1}{(n+2)/2}=2-\frac{2}{n+2}, & \ceil*{ \frac{n+1}{R+1}} &= 2, & \frac{R+\ceil{d/2}}{R+1}&= \frac{n}{n/2+1}=2-\frac{4}{n+2}.
\end{align*}
Starting from the left-hand side of~\eqref{eq:newdef}, we get:
\[
   2\cdot \parenv*{\sum_{i=0}^{n/2}\binom{n}{i}-\binom{n}{n/2} \frac{\frac{2}{n+2}}{\frac{4}{n+2}}}=2\cdot \parenv*{\frac{1}{2}\parenv*{2^n-\binom{n}{n/2}}+\binom{n}{n/2}-\frac{1}{2}\binom{n}{n/2}}=2^n
   \]
\end{example}

\begin{example}\label{ex:len6}
Consider the following code:
\[C=\set*{000000,000111,111000,111111}.\]
One can easily verify the code is a $(6,4,3)2$ code. Plugging in these values into the left-hand side of~\eqref{eq:newdef} yields $2^6$, hence the code is a nearly-perfect covering code. This code is depicted in Figure~\ref{fig:len6}. Vectors are abbreviated using two digits, so that $ij$ indicates any vector which has weight $i$ in the first three coordinates, weight $j$ in the last three. For example, $02$ denotes any of the vectors $000110$, $000101$, or $000011$. In this notation, $C=\set{00,03,30,33}$. The packing-radius balls around the codewords are dashed, while the covering-radius balls around the codewords are solid.

Additionally, consider the following code:
\[
C'=\set*{000000,000111,111001,111110}.
\]
This is also a $(6,4,3)2$ nearly-perfect covering code (verification by inspection is straightforward). One can easily see that $C$ and $C'$ are not equivalent codes since the former has two codewords at distance $6$ from each other while the latter does not.
\end{example}

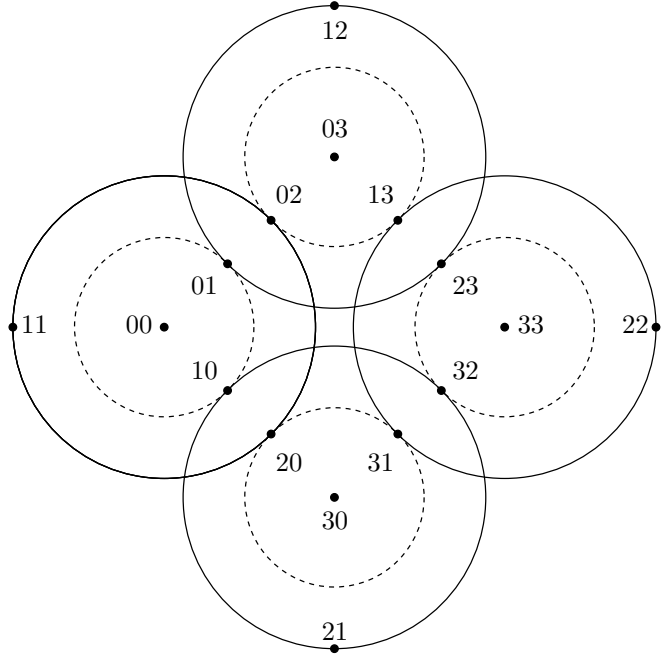
\begin{figure}
\begin{center}
    \begin{overpic}[scale=0.25]
        {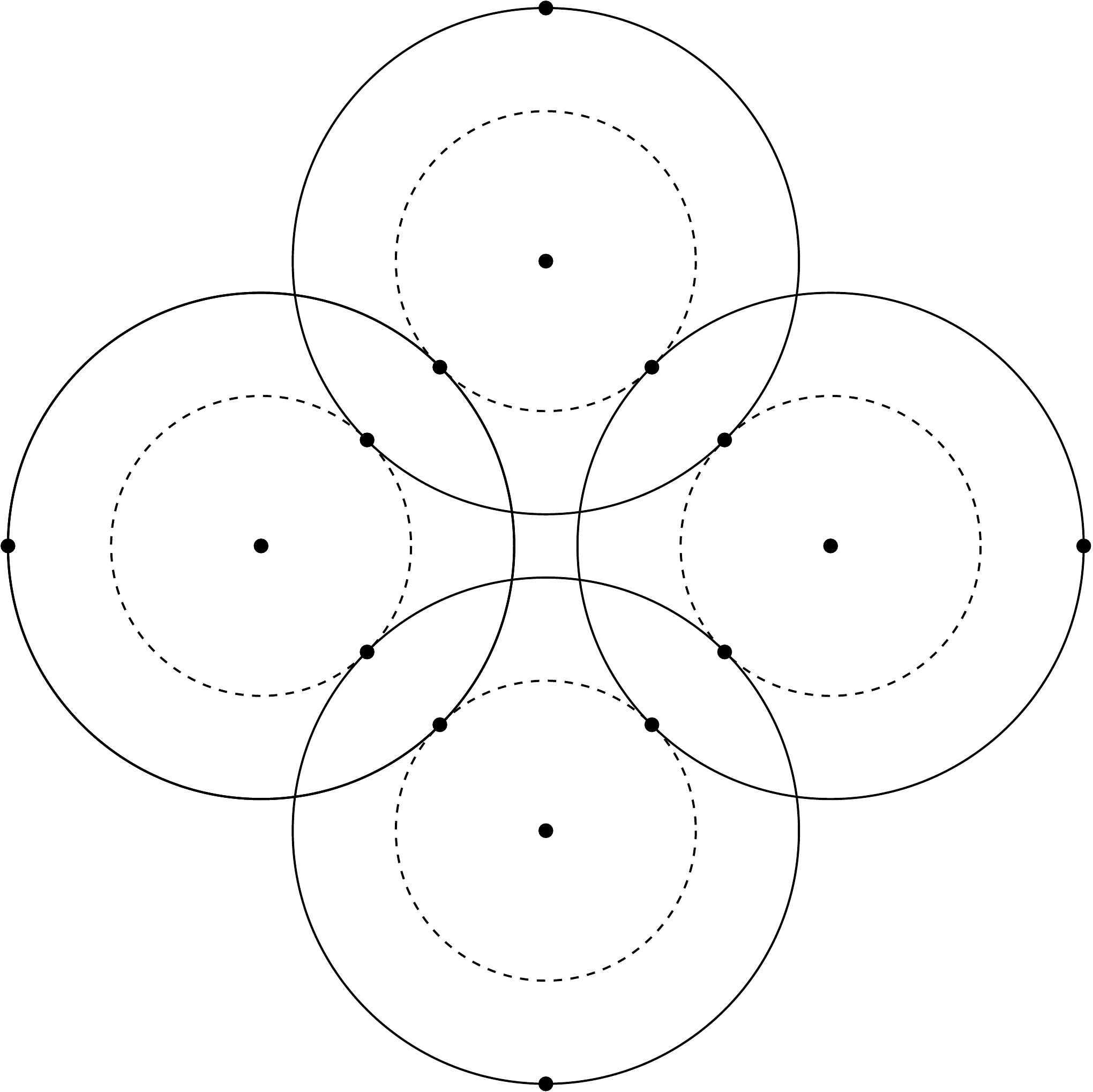}
        \put(18,49){$00$}
        \put(78,49){$33$}
        \put(48,79){$03$}
        \put(48,19){$30$}
        \put(2,49){$11$}
        \put(94,49){$22$}
        \put(48,94){$12$}
        \put(48,2){$21$}
        \put(28,55){$01$}
        \put(28,42){$10$}
        \put(68,55){$23$}
        \put(68,42){$32$}
        \put(41,69){$02$}
        \put(55,69){$13$}
        \put(41,28){$20$}
        \put(55,28){$31$}
    \end{overpic}
\end{center}
\caption{A depiction of the code $C$ from Example~\ref{ex:len6}.}
\label{fig:len6}
\end{figure}

\section{Properties of Nearly-Perfect Covering Codes}
\label{sec:properties}

In this section we study the structural properties of nearly-perfect covering codes. These will be used later to rule out the existence of such codes for various sets of parameters. The main approach we use is to analyze the consequences of having~\eqref{ineq2}-\eqref{ineq7} as equalities.

First, we show that all nearly-perfect codes are quasi-perfect. Recall that quasi-perfect codes are those codes for which the difference between the covering radius and packing radius is at most $1$.

\begin{lemma}\label{lemma:quasiperfect}
A nearly-perfect covering code is a quasi-perfect code, and thus $d\in\set{2R-1,2R}$.
\end{lemma}

\begin{proof}
Since~\eqref{ineq6} has to be an equality, we have $\abs{C_R}=2^n-M\sum_{i=0}^{R-1}\binom{n}{i}$, which implies that the balls of radius $R-1$ centered at the codewords, are disjoint. Thus, the packing radius is at least $R-1$, and no more than the covering radius, $R$, so the code is quasi-perfect. Since by definition, a nearly-perfect covering code is non-perfect, the packing radius must be exactly $R-1$. Hence, $\floor{(d-1)/2}=R-1$, which proves that $d\in\set{2R-1,2R}$.
\end{proof}

It is interesting to note that nearly-perfect \emph{error-correcting} codes are also always quasi-perfect. Recall that these are only $(2^m-2, 2^{2^m-2-m},3)2$  and $(2^m-1,2^{2^m-2m},5)3$ codes (e.g., see~\cite[Section 2.3.5]{PleHuf03}). By plugging their parameters into~\eqref{eq:newdef}, we find that none of them are nearly-perfect covering. Thus, nearly-perfect error-correcting codes and nearly-perfect covering codes are disjoint subsets of the quasi-perfect codes.

In what follows, we will frequently need to deal with the quantity $\ceil{(n+1)/(R+1)}$. For convenience, we denote
\begin{align*}
a&\eqdef\ceil*{\frac{n+1}{R+1}}, \\
b&\eqdef (R+1)a-(n+1) = (R+1)\ceil*{\frac{n+1}{R+1}}-(n+1),
\end{align*}
so $0\leq b \leq R$. Now, using Lemma~\ref{lemma:quasiperfect}, we can plug in $\ceil{d/2}=R$ into~\eqref{eq:newdef} to get a more compact form:

\begin{corollary}\label{coro:defwith2Rreplaced}
An $(n,M,d)R$ nearly-perfect covering code satisfies:
    \begin{equation}\label{eq:newdefevaluatedd}
    M\parenv*{\sum_{i=0}^R\binom{n}{i}-\binom{n}{R}\frac{b}{n+1+b-2R}}=2^n,    
    \end{equation}
    and additionally, $b\neq 0$.
\end{corollary}

We note that when $b=0$, \eqref{eq:newdefevaluatedd} becomes the ball-packing bound, and codes attaining it with equality are perfect, so by definition, not nearly-perfect.The last corollary is particularly useful when ruling out parameter sets for nearly-perfect codes using number-theoretic tools. Next, we turn to study the properties of over-covered vectors.

\begin{lemma}\label{lemma:epsilonequality}
In a nearly-perfect covering code, for any $z\in C_R$
\begin{align*}
g(B_1(z))=(R+1)\parenv*{\ceil*{\frac{n+1}{R+1}}-\frac{n+1}{R+1}}&=b,\\
\abs*{B_{R+1}(z)\cap C}=\ceil*{\frac{n+1}{R+1}}&=a.
\end{align*}
\end{lemma}

\begin{proof}
\eqref{ineq7} has to be an equality, and thus \eqref{eq:ineqepsilon} in Lemma~\ref{lemma:epsilon} has to be an equality for every $x\in C_R$.
\end{proof}

We remark that by instantiating Lemma~\ref{lemma:epsilonequality} with $R=1$ we obtain \cite[Theorem 3]{BorEtzRot25} as a Corollary.

\begin{lemma}
\label{lemma:muequalitycond}
If $C$ is an $(n,M,d)R$ nearly-perfect covering code, and $z\in\F_2^n$ is over-covered, then it is covered by exactly two codewords, $c$ and $c'$, and:
    \begin{itemize}
        \item If $d$ is even, then $d(z,C)=R=d(z,c)=d(z,c')$, and $d(c,c')=d$.
        \item If $d$ is odd and $d(z,C)=R$, then $d(z,c)=d(z,c')=R$ and $d(c,c')=d+1$.
        \item If $d$ is odd and $d(z,C)=R-1$, then $d(z,c)=R$, $d(z,c')=R-1$, and $d(c,c')=d$.
    \end{itemize}
\end{lemma}

\begin{proof}
Since~\eqref{ineq2} has to be an equality, $\mu=n+1-R-\ceil{d/2}$, and since \eqref{ineq4} has to be an equality, $\abs{B_1(z)\cap C_R}=\mu$ for every over-covered $z$. We employ the same notations as in the proof of Lemma~\ref{lemma:muineq}. We deduce from the proof of Lemma~\ref{lemma:muineq} that if $d(z,C)=R$, then $\abs{\supp(\widetilde{T}(z))}=R+\ceil{d/2}$, and if $d(z,C)=R-1$, then $\abs{\supp(\widetilde{T}(z))}=R-1+\ceil{d/2}$. Also, we need to have $\abs{\supp(\widetilde{T}(z))}=\abs{A\cup A'}$, and thus $\supp(\widetilde{T}(z))= A\cup A'$. 

We note that if $z$ is over-covered, then $d(z,C)=R$ or $R-1$ since the code is quasi-perfect. Now we split the analysis according to $d(z,C)$ and the parity of $d$.
    
\textbf{Case 1:} $d$ even and $d(z,C)=R$. In this case, $d(z,c)=d(z,c')=R$, and from the proof of Case 1 in Lemma~\ref{lemma:muineq} we see that to have equality we need $\abs{\supp (\widetilde{T}(z))}=R+d/2$, so $\abs{A\cap A'}=R-d/2$, hence 
\[
d(c,c')=d(y,y')=2R-2(R-d/2)=d,
\]
as we wanted. If there were a third codeword $c''$ covering $z$, where $\supp (y'')=\supp (z-c'')=A''$, by repeating the same argument as before we would need to have $\supp (\widetilde{T}(z))=A\cup A'= A'\cup A''=A''\cup A$, so $A''\subseteq A\cup A'$, and similarly $A\subseteq A'\cup A''$, $A'\subseteq A\cup A''$. Define the following quantities,
\begin{align*}
s_1 &\eqdef \abs*{(A\cap A')\setminus A''}, &
s_2 &\eqdef \abs*{(A\cap A'')\setminus A'}, \\
s_3 &\eqdef \abs*{(A'\cap A'')\setminus A}, &
s_4 &\eqdef \abs*{A\cap A' \cap A''}.
\end{align*}
For ease of presentation, the sizes of the relevant subsets are presented in the following Venn diagram:
\begin{center}
\begin{overpic}[scale=0.4]
    {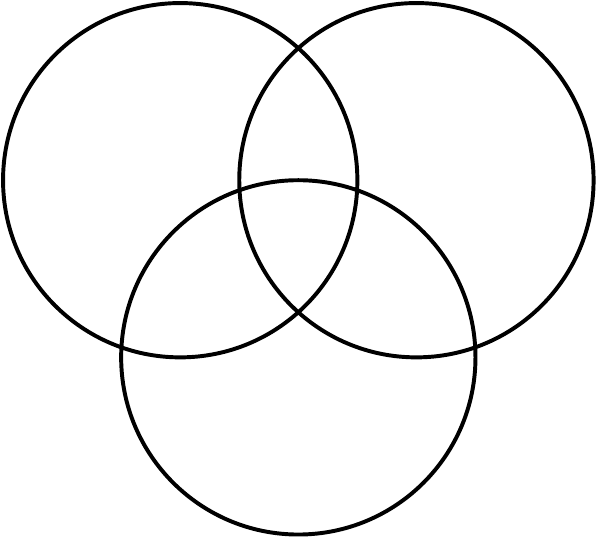}
    \put(-10,58){$A$}
    \put(102,58){$A'$}
    \put(14,5){$A''$}
    \put(46,65){$s_1$}
    \put(31,40){$s_2$}
    \put(62,40){$s_3$}
    \put(46,49){$s_4$}
    \put(20,65){$0$}
    \put(75,65){$0$}
    \put(48,15){$0$}
\end{overpic}
\end{center}
We now have the following equations:
\begin{align*}
s_1+s_4 = \abs*{A\cap A'} &= R-d/2, &
s_1+s_2+s_4 = \abs{A} &= R, \\
s_2+s_4 = \abs*{A\cap A''} &= R-d/2, &
s_1+s_3+s_4 = \abs{A'} &= R, \\
s_3+s_4 = \abs*{A'\cap A''} &= R-d/2, &
s_2+s_3+s_4 = \abs{A''} &= R. \\
\end{align*}
Solving these equations yields $s_4=R-d$. Because $d$ is even, by Lemma~\ref{lemma:quasiperfect}, $d=2R$ and $s_4=-R<0$, a contradiction. It follows that $z$ cannot be over-covered by three codewords.

\textbf{Case 2:} $d$ even and $d(z,C)=R-1$. Since the code is quasi-perfect, we cannot have both $c$ and $c'$ at distance $R-1$ from $z$, so we may suppose $d(z,c)=R$ and $d(z,c')=R-1$. By the proof of Case 2.b in Lemma~\ref{lemma:muineq}, in this case $\abs{\supp (\widetilde{T}(z))}\geq R-1+\ceil{(d+1)/2}\geq R+d/2$, since $d$ is even. But in order to have equality, we need $\abs{\supp (\widetilde{T}(z))}=R-1+d/2$, a contradiction. This case is thus not viable. 

\textbf{Case 3:} $d$ odd and $d(z,C)=R$. This time we need $\abs{\supp (\widetilde{T}(z))}=R+\ceil{d/2}=R+(d+1)/2$, and $\abs{A\cap A'}=R-(d+1)/2$, so
\[
d(c,c')=d(y,y')=\abs{A}+\abs{A'}-2\abs{A\cap A'}=2R-(2R-(d+1))=d+1,
\]
as claimed. 

If there were a third codeword $c''$ covering $z$, with the same notations as before, we would have:
\begin{align*}
s_1+s_4 = \abs*{A\cap A'} &= R-(d+1)/2, &
s_1+s_2+s_4 = \abs{A} &= R, \\
s_2+s_4 = \abs*{A\cap A''} &= R-(d+1)/2, &
s_1+s_3+s_4 = \abs{A'} &= R, \\
s_3+s_4 = \abs*{A'\cap A''} &= R-(d+1)/2, &
s_2+s_3+s_4 = \abs{A''} &= R. \\
\end{align*}
Solving these equations yields $s_4=R-(d+1)= R-2R<0$, a contradiction. So $z$ cannot be covered by more than two codewords.

\textbf{Case 4:} $d$ odd and $d(z,C)=R-1$. Since the code is quasi-perfect, we have $d(z,c)=R$ and $d(z,c')=R-1$. By the proof of Case 2.b in Lemma~\ref{lemma:muineq}, $\abs{\supp (\widetilde{T}(z))}=R-1+(d+1)/2$ and $\abs{A\cap A'}=R-(d+1)/2$. Thus, 
\[
d(c,c')=d(y_1,y_2)=\abs{A}+\abs{A'}-2\abs{A\cap A'}=2R-1-(2R-(d+1))=d.
\]

Suppose there were a third codeword $c''$ covering $z$. We know it cannot be at distance $R-1$ from $z$ since $c'$ already is, so $d(c'',z)=R$. But if there were two codewords at distance $R$ from $z$, then by the computations in Case 1 of the proof of Lemma \ref{lemma:muineq} we would have $\abs{\supp (\widetilde{T}(z))}\geq R+ \ceil{d/2}$, a contradiction, since we need $\abs{\supp (\widetilde{T}(z))}=R-1+\ceil{d/2}$.    
\end{proof}

\section{The Case of $R=2$}
\label{sec:R2}

In~\cite{BorEtzRot25}, nearly-perfect covering codes with $R=1$ were studied. It was proved there that their parameters are exactly $(2^m, 2^{2^m-m},2)1$ or $(2^m,2^{2^m-m},1)1$. With our results above, we extend the treatment beyond $R=1$. In this section we prove no non-perfect $(n,M,d)2$ code exists, except in the parameters of Example~\ref{ex:repeat} and Example~\ref{ex:len6}. Additionally, any such code must be equivalent to the list of four codes given in these examples. 

\subsection{Parameter Classification}

We first prove the claim about the possible code parameters. To do so, we need two technical lemmas:

\begin{lemma}
\label{lem:no6442code}
There is no $(6,4,4)2$ code.
\end{lemma}
\begin{proof}
Assume to the contrary that such a code $C$ exists. By translating the code, we may assume, without loss of generality, that it contains the codeword $0$, i.e., $C=\set{0,x,y,z}$, with $x,y,z\in\F_2^6$. By the minimum distance of the code, $\wt(x),\wt(y),\wt(z)\geq 4$. Let $\overline{x},\overline{y},\overline{z}$ be the binary complements of $x,y,z$, respectively. Thus, $\wt(\overline{x}),\wt(\overline{y}),\wt(\overline{z})\leq 2$. Complementing the vectors is translating them (adding the all-ones vector), and so the distances are preserved. Thus, $d(\overline{x},\overline{y}),d(\overline{x},\overline{z}),d(\overline{y},\overline{z})\geq 4$, to satisfy the minimum distance of the code. Because of the complements' weights, this is only possible if
\[
d(\overline{x},\overline{y}),d(\overline{x},\overline{z}),d(\overline{y},\overline{z})= 4,
\]
and their supports are pairwise disjoint. It follows that, up to a permutation of coordinates,
\[
C=\set{000000,111100,110011,001111}.
\]

None of the above assumptions (permutation and translation) change the covering radius of the code $C$. However, $C$ has covering radius at least $3$, since the distance of the vector $010101$ from all the codewords is
\[
d(010101,000000)=d(010101,111100)=d(010101,110011)=d(010101,001111)=3.
\]
\end{proof}

\begin{lemma}
\label{lem:dio}
The solutions of the Diophantine equation
\[
n^3-2n^2+n-4=2^z,
\]
are $(n,z)=(3,3)$ and $(n,z)=(4,5)$.
\end{lemma}
\begin{proof}
We divide our analysis into two cases, depending on the parity of $z$.

\textbf{Case 1:} $z$ is even. Define $x=2^{z/2}$, and the equation becomes
\[
n^3-2n^2+n-4=x^2.
\]
This is a standard (Weierstrass form) elliptic curve. The integral solutions of such an equation may be found algorithmically (e.g., see~\cite[Chapter IX]{Sil86}). Using such an algorithm\footnote{We used SageMath.} we find this equation has no integral solutions.

\textbf{Case 2:} $z$ is odd. Define $x=2^{(z-1)/2}$, and the equation becomes
\begin{equation}
\label{eq:dio1}
n^3-2n^2+n-4=2x^2.
\end{equation}
To arrive at a Weierstrass form, we multiply both sides by $8$, and define $N=2n$ and $X=4x$. With this notation the equation is
\begin{equation}
\label{eq:dio2}
N^3-4N^2+4N-32=X^2,
\end{equation}
and crucially, all integral solutions of~\eqref{eq:dio1} correspond to integral solutions of~\eqref{eq:dio2} (but not necessarily vice versa). Since~\eqref{eq:dio2} is an elliptic curve, we find its integral points, as before, and obtain the following solutions:
\[
(N,X) \in \set*{ (6,\pm 8), (8,\pm 16), (134,\pm 1528)},
\]
which imply the following solutions to~\eqref{eq:dio1}:
\[
(n,x) \in \set*{ (3,\pm 2), (4,\pm 4), (67,\pm 382)}.
\]
Solutions with negative $x$ are not relevant, since $x=2^{(z-1)/2}$, and so we remove three of six solutions. Of the remaining three, $(n,x)=(67,382)$ is impossible since $382$ is not a power of $2$. The remaining two solutions may be translated back to solutions of the original equation, in terms of $n$ and $z$, giving
\[
(n,z)\in\set*{(3,3),(4,5)},
\]
as claimed.
\end{proof}

We are now in a position to state and prove the main claim about the code parameters.

\begin{theorem}\label{thm:radius2}
When the covering radius is $R=2$, non-perfect nearly-perfect covering codes exist only with parameters $(4,2,3)2$, $(4,2,4)2$, or $(6,4,3)2$.
\end{theorem}

\begin{proof}
Consider a non-perfect $(n,M,d)2$ nearly-perfect covering code. By Lemma~\ref{lemma:quasiperfect}, $d\in\set{3,4}$. By Corollary \ref{coro:defwith2Rreplaced}, we have,
    \begin{equation}\label{eq:proofR2}
         M \parenv*{1+n+\binom{n}{2}-\binom{n}{2}\frac{b}{n+b-3}} = 2^n.
    \end{equation}
where $b\in\set{1,2}$ since we are interested in non-perfect codes. We separate the analysis into two cases according to the value of $b$.

\textbf{Case 1:} $b=2$. Equation~\eqref{eq:proofR2} becomes
\[
M\parenv*{1+\binom{n}{2}}=2^n.
\]
Since both $M$ and $1+\binom{n}{2}$ are integers, they both have to be a power of $2$. We therefore have to solve
\[
n^2-n+2 = 2^z,
\]
for some non-negative integer $z$. Viewing this as a quadratic equation in $n$, the discriminant $(-1)^2-4(2-2^z)=2^{z+2}-7$ must be a square. Denote $x=z+2$, and hence, we need to solve the Diophantine equation
\[
2^x-7=y^2,
\]
for the unknowns $x$ and $y$. This is the Ramanujan-Nagell equation, whose only solutions are known to be $(x,y)=(3,\pm 1)$,$(4,\pm 3)$, $(5,\pm 5)$, $(7,\pm 11)$, $(15,\pm 181)$ (see~\cite{Nag60}). These solutions yield the values $n=0,1,2,3,6,91$. Recall that $b=2$, which implies $n\equiv 0\pmod 3$, so that the only possible values are $3$ and $6$. 
    
If $n=3$, \eqref{eq:proofR2} implies $M=2$. The minimum distance has to be $d=3$, and these are the parameters of the binary repetition code of length $3$, which is perfect, and thus not a solution. If $n=6$, then $M=4$, in which case, by Lemma~\ref{lemma:quasiperfect}, we have two options: either $d=3$ or $d=4$. For the first option, a $(6,4,3)2$ nearly-perfect covering code indeed exists, as seen in Example~\ref{ex:len6}. For the second option, a $(6,4,4)2$ code does not exist by Lemma~\ref{lem:no6442code}.

\textbf{Case 2:} $b=1$. In this case, we arrange~\eqref{eq:proofR2} and obtain
\begin{equation}
\label{eq:proofR2b1}
\frac{M(n^3-2n^2+n-4)}{n-2}=2^{n+1}.
\end{equation}
The right-hand side is a power of $2$. On the left-hand side, we observe that
\[
(n^3-2n^2+n-4) - (n^2+1)(n-2) = -2,
\]
so
\[
\gcd(n^3-2n^2+n-4,n-2) \in \set{1,2}.
\]
Thus, whatever prime factors different from $2$ (if any) exist in $n-2$, must be canceled by $M$ in~\eqref{eq:proofR2b1}, and $n^3-2n^2+n-4$ must be a power of $2$. Hence, we need to solve the Diophantine equation
\[
n^3-2n^2+n-4 = 2^z,
\]
for some non-negative integer $z$. By Lemma~\ref{lem:dio}, the solutions are $(n,z)=(3,3)$ and $(n,z)=(4,5)$. We now recall that $b=1$, and therefore $n\equiv 1\pmod{3}$. Thus, the only solution is $n=4$.

If $n=4$, by~\eqref{eq:proofR2} we have $M=2$. Since the minimum distance is $d=3$ or $d=4$ (by Lemma~\ref{lemma:quasiperfect}), the potential nearly-perfect covering codes have parameters $(4,2,3)2$ and $(4,2,4)2$. Both of these codes exist, as shown in Example~\ref{ex:repeat}.
\end{proof}

\subsection{Code Equivalence}

Indeed, the codes of Example~\ref{ex:repeat} and Example~\ref{ex:len6} cover all possible parameters for a nearly-perfect covering code of radius $R=2$ given in Theorem~\ref{thm:radius2}. We now show all such codes must be equivalent to one of the codes in the examples. To make our statements precise, we say two codes are equivalent if one may be obtained from the other by a permutation of the coordinates and a translation by some fixed vector.

\begin{theorem}
\label{th:repuniq}
If a non-perfect $(n,2,d)R$ code, $C$, is nearly perfect, it is equivalent to one of the codes from Example~\ref{ex:repeat}.
\end{theorem}

\begin{proof}
Translate the code $C$ so that it contains the all-zero vector, i.e., it becomes $\set{0,c}$ for some $c\in\F_2^n$. It is easy to verify that the covering radius of the code is exactly
\begin{equation}
\label{eq:rwtc}
R=n-\ceil*{\frac{\wt(c)}{2}}.
\end{equation}
The minimum distance of the code is therefore $d=\wt(c)$. By Lemma~\ref{lemma:quasiperfect}, the code is quasi-perfect and so \[\wt(c)=d\in\set{2R-1,2R}.\]
This implies $\ceil{\frac{\wt(c)}{2}}=R$. Plugging this in~\eqref{eq:rwtc} we obtain $n=2R$. The code is therefore either $(2R,2,2R-1)R$ or $(2R,2,2R)R$, and after a proper permutation, must be equal to one of the codes of Example~\ref{ex:repeat}, as claimed.
\end{proof}

The next theorem examines remaining possible parameters:

\begin{theorem}
If $C$ is a $(6,4,3)2$ nearly-perfect code, then it is equivalent to one of the codes from Example~\ref{ex:len6}.
\end{theorem}

\begin{proof}
After translation, we may assume without loss of generality, that $C$ contains the all-zero codeword, i.e., $C=\set{0,c,c',c''}$. A ball of radius $R=2$ in $\F_2^6$ has size $1+\binom{6}{1}+\binom{6}{2}=22$. The over-covering of the space is
\[
g(\F_2^6) = 4\cdot 22 - 2^6=24.
\]
By Lemma~\ref{lemma:muequalitycond}, an over-covered vector is covered by exactly two codewords. Thus, the number of over-covered vectors is exactly $24$. A simple calculation shows that two balls of radius $2$ whose centers are at distance $3$ or $4$, have an intersection of size $6$. If, however, their centers are at distance $5$ or $6$, then they are disjoint.

The four codewords give rise to six unordered pairs of codewords. By the previous calculation, exactly four out of the six pairs must intersect in six vectors, giving rise to the over-covering of $24$ vectors.

The remaining two codeword pairs must therefore describe codewords at distance $5$ or $6$ apart. We first contend that these two pairs cannot share a codeword, i.e., cannot be of the form $\set{w,w'},\set{w,w''}$ with $w,w',w''\in C$. If that were the case, and because the length of the code is $n=6$, then $d(w',w'')\leq 2$, contradicting the minimum distance of the code. Thus, the two pairs contain all  four codewords, and are, without loss of generality, $\set{0,c}$ and $\set{c',c''}$. We summarize our findings so far by:
\begin{align*}
d(0,c),d(c',c'')&\in\set*{5,6}, \\
d(0,c'),d(0,c''),d(c,c'),d(c,c'')&\in\set*{3,4}.
\end{align*}
We now distinguish between two cases depending on $\wt(c)=d(0,c)$.

\textbf{Case 1:} $\wt(c)=6$. Thus $c=111111$. Since $\wt(c')=d(0,c')\in\set{3,4}$, but also $d(c,c')\in\set{3,4}$, the only possible option is having $\wt(c')=3$. A similar argument gives $\wt(c'')=3$. Finally, since $d(c',c'')\in\set{5,6}$, and since we must have an even distance between two equal-weight vectors, then $d(c',c'')=6$ and the two codewords are binary complements of each other. After an appropriate permutation, $c'=111000$ and $c''=000111$, giving the code $C$ from Example~\ref{ex:len6}.

\textbf{Case 2:} $\wt(c)=5$. We know that $\wt(c'),\wt(c'')\in\set{3,4}$. We cannot have $\wt(c')=\wt(c'')=3$, because then we would need to have $d(c',c'')=6$ (i.e., they are binary complements of each other) and one of $c'$ or $c''$ will have distance at most $2$ from $c$, contradicting the minimum distance of the code. We also cannot have $\wt(c')=\wt(c'')=4$ because then $d(c',c'')\leq 4$. Thus, without loss of generality, $\wt(c')=3$ and $\wt(c'')=4$. Necessarily $d(c',c'')=5$ and they have exactly one common position containing a $1$.

Without loss of generality, $c=111110$. To avoid contradicting the minimum distance of the code, $c'$ must have a $1$ in the last coordinate. By permuting only the first five positions, we can assume $c'=000111$. This leaves a single option for $c''$ to have distance at least $3$ from $c$, which is $c''=111001$. The code is then $C'$ from Example~\ref{ex:len6}.
\end{proof}

\section{Higher Radii $R\geq 3$}
\label{sec:R3ormore}

In this section we show how to extend the number-theoretic method used in Theorem~\ref{thm:radius2} for $R=2$, to the case of $R\geq 3$. The main obstacle is that the elliptic curve equations in Theorem~\ref{thm:radius2} become higher degree equations and are no longer elliptic curves. These, however, may be solved through the method of Thue-Mahler equations (e.g., see~\cite[Chapter VIII]{Sma98}). This algorithmic way may be applied to any fixed $R\geq 3$. In what follows, we give the proof for $R=3$, and write in a general way, leaving the specialization to the case $R=3$ to the very end.

We first recall the definition of a Thue-Mahler equation. We say $F(x,y)$ is a homogeneous integer polynomial of degree $d$ in the unknowns $x$ and $y$ if,
\[
F(x,y) = \sum_{i=0}^d a_i x^i y^{d-i},
\]
with $a_i\in \Z$, not all of them zero. A Thue-Mahler Diophantine equation is an equation of the following form:
\[
F(x,y) = C p_1^{z_1}p_2^{z_2}\dots p_t^{z_t},
\]
where $F(x,y)$ is a homogeneous integer polynomial, $C\in\Z$ is some constant, and $p_1,\dots,p_t$ are distinct primes. The unknowns we solve for are the integers $x$, $y$, as well as $z_1,\dots,z_t$. When $d\geq 3$ and $F(x,y)$ is irreducible over $\Q$, the number of solutions is finite~\cite{Thu09,Mah33}, and those may be found algorithmically. For a full description, see~\cite[Chapter VIII]{Sma98}, and the references therein.

We also need the definition of a falling factorial. We denote by $(x)_\ell$, $\ell\in\Z$, $\ell\geq 0$, the expression
\[
(x)_\ell \eqdef x(x-1)(x-2)\dots(x-\ell+1).
\]
In particular $(x)_0=1$.

\begin{theorem}
\label{th:codesR3}
When the covering radius is $R=3$, non-perfect nearly-perfect covering codes exist only with parameters $(6,2,5)3$, or $(6,2,6)3$.
\end{theorem}

\begin{proof}
Assume an $(n,M,d)R$ nearly-perfect covering code exists. We rewrite Corollary~\ref{coro:defwith2Rreplaced} and obtain
\begin{equation}
\label{eq:geneq}
\frac{M\parenv*{(n-(2R-b-1))\sum_{i=0}^R (n)_i (R)_{R-i}-b\cdot(n)_R}}{R!(n-(2R-b-1))}=2^n.
\end{equation}
Since $b=0$ corresponds to perfect codes, we only need to check all $1\leq b\leq R$. Due to some cancellation, we leave the case of $b=R$ for later. For $1\leq b\leq R-1$, define the integer polynomial
\[
F_b(n) \eqdef (n-(2R-b-1))\sum_{i=0}^R (n)_i (R)_{R-i}-b\cdot(n)_R.
\]
Since the denominator in~\eqref{eq:geneq} is linear in $n$,
\[
F_b(n) \equiv F_b(2R-b-1) = -b\cdot(2R-b-1)_R \pmod{n-(2R-b-1)},
\]
which does not depend on $n$. Hence,
\[
\gcd(F_b(n),n-(2R-b-1)) \mid b\cdot(2R-b-1)_R.
\]
Assume the prime factorization of $R!\cdot b\cdot(2R-1-b)_R$ is $p_1^{s_1}\dots p_t^{s_t}$, where $p_1<\dots< p_t$ are distinct primes and $s_1,\dots,s_t\geq 1$, and notice that $p_1=2$ when $R\geq 2$. Since the right-hand side of~\eqref{eq:geneq} contains only $2$ in the prime factorization, the prime factors of $F_b(n)$ must be a subset of $\set{p_1,\dots,p_t}$. Thus, by~\eqref{eq:geneq},
\begin{equation}
\label{eq:thuemahler1}
F_b(n) = p_1^{m_1}p_2^{m_2}\dots p_t^{m_t},
\end{equation}
for some non-negative integers $m_i$. Moreover, $m_i\leq s_i$ for each $2\leq i\leq t$. That is, the exponent of every prime besides $2$ can be at most the exponent in the factorization of the denominator.

For the case $b=R$, the term $n-(2R-b-1)=n-(R-1)$, gets canceled from the numerator and denominator of~\eqref{eq:geneq}, and we are left with the equation
\[
\frac{M\parenv*{\sum_{i=0}^R (n)_i (R)_{R-i}-R\cdot(n)_{R-1}}}{R!}=2^n.
\]
In the same fashion as above, we define
\[
F_R(n)=\sum_{i=0}^R (n)_i (R)_{R-i}-R\cdot(n)_{R-1},
\]
and we let $p_1^{s_1}\dots p_t^{s_t}$ be the prime factorization of $R!$, where $2=p_1<\dots<p_t$ are distinct primes and $s_1,\dots,s_t\geq 1$. Thus, we need to solve
\begin{equation}
\label{eq:thuemahler1R}
F_R(n) = p_1^{m_1}p_2^{m_2}\dots p_t^{m_t},
\end{equation}
subject to $m_i\leq s_i$ for every $2\leq i\leq t$.

Let $F_b(n,m)$ be the homogeneous polynomial such that $F_b(n,1)=F_b(n)$, for all $1\leq b\leq R$. Then every integer solution $n,m_1,\dots,m_t$ to~\eqref{eq:thuemahler1}, or~\eqref{eq:thuemahler1R} when $b=R$, induces an integer solution to
\begin{equation}
\label{eq:tmneed}
F_b(n,m) = p_1^{m_1}p_2^{m_2}\dots p_t^{m_t},
\end{equation}
that is $n,m=1,m_1,\dots,m_t$. It is now a matter of checking all the integral solutions of this Thue-Mahler equation, and verifying which ones fit the original equation~\eqref{eq:geneq}\footnote{If the bounds on the exponents $m_2,\dots, m_t$ are small, it could be more efficient to solve each of the $\prod_{i=2}^{t}(s_i+1)$ Thue-Mahler equations $F_b(n,m)=C 2^{m_1}$, with the constants $C=p_2^{u_2}\dots p_t^{u_t}$ for each choice of $0\leq u_2\leq s_2, \dots, 0\leq u_t\leq s_t$.}.

At this point we specialize to the case of $R=3$. We divide the analysis according to the value of $1\leq b\leq 3$.

\textbf{Case 1:} $b=1$. We have $R!b\cdot(2R-b-1)_R=144=2^4 3^2$. We set $p_1=2$ and $p_2=3$. We need to solve the Thue-Mahler equation
\[
F_1(n,m)=n^4 -5n^3 m +8 n^2 m^2- 16n m^3 - 24m^4 = 2^{m_1}3^{m_2}.
\]

The solutions we get are\footnote{We solved the Thue-Mahler equation using magma and the Thue-Mahler solver~\cite{GheSik25}, available online at \url{https://github.com/adelagherga/ThueMahler/tree/master/Code/TMSolver}. The solver also checks irreducibility.} 
\begin{align*}
(n,m,m_1,m_2)\in\{&(\pm 82,\pm 19,9,5),(\pm 6,\pm 1,7,1), (\pm 5, \pm 1, 5, 1),(\pm 2,\mp 1, 5, 1), \\
&(\pm 1,\mp 1, 1, 1),(\pm 1,0,0,0)\}.
\end{align*}
These solutions contain only those in which $\gcd(n,m)=1$ for two reasons. First, any solution $(n,m)$ gives rise to infinitely many other $(cn,cm)$ for any integer $c$ that contains only $p_1,\dots,p_t$ as prime factors (with a proper adjustment of $m_1,\dots,m_t$). Thus, it is enough to specify those solutions with $\gcd(n,m)=1$. Second, we used $m$ artificially to get to a homogeneous polynomial, and are therefore only interested in solutions in which $m=1$.

We remove all solutions where $n< 2R-1$ (since these are obviously irrelevant), and those with $m\neq 1$. We also need to remove all solutions with $m_2>2$. Additionally, since $b=1$ we must have $n\equiv 2\pmod{4}$. After removing all irrelevant solutions we are left with a single solution, $n=6$. Checking~\eqref{eq:geneq}, this solution corresponds to either a $(6,2,5)3$ code, or a $(6,2,6)3$ code, as claimed.

\textbf{Case 2:} $b=2$. In this case we have $R!(2R-b-1)_R = 36=2^2 3^2$, so again, $p_1=2$ and $p_2=3$. We therefore need to solve
\[
F_2(n,m)=n^4 -5n^3 m + 11n^2 m^2- 13n m^3 - 18m^4 = 2^{m_1}3^{m_2}.
\]
The solutions we get are
\begin{align*}
(n,m,m_1,m_2) \in \{ &
(\pm 18, \pm 5, 5, 5),
(\pm 17, \pm 5, 4, 4),
(\pm 9, \mp 11, 6, 6),
(\pm 7, \mp 5, 13, 1), \\
&
(\pm 5, \pm 1, 6, 1),
(\pm 2, \mp 1, 2, 3),
(\pm 1, \mp 1, 2, 1),
(\pm 1, 0, 0, 0)
\}.
\end{align*}
After removing irrelevant solutions (this time we must have $n\equiv 1\pmod{4}$), the only remaining candidate is $n=5$. However, plugging this into~\eqref{eq:geneq} gives us the parameters of the $(5,2,5)2$ code which is perfect.

\textbf{Case 3:} $b=3$. In this final case, $R!=6=2\cdot 3$, and as before, $p_1=2$ and $p_2=3$. The equation we need to solve is
\[
F_3(n,m)=n^3-3n^2 m+8n m^2+6m^3 = 2^{m_1}3^{m_2}.
\]
The solutions we get are
\begin{align*}
(n,m,m_1,m_2) \in \{ &
    ( 85, -143, 13, 2 ),
    ( 59, 7, 3, 9 ),
    ( 28, -47, 1, 7 ),
    ( 11, -17, 7, 3 ),
    ( 7, -5, 6, 3 ), \\
&
    ( 6, 1, 1, 4 ),
    ( 5, 1, 5, 1 ),
    ( 4, 1, 1, 3 ),
    ( 3, 7, 10, 1 ),
    ( 3, -1, 3, 2 ), \\
&
    ( 3, -5, 2, 1 ),
    ( 2, 1, 1, 2 ),
    ( 1, 1, 2, 1 ),
    ( 1, -1, 1, 1 ),
    ( 1, 0, 0, 0 ),
    ( 0, 1, 1, 1 ), \\
&
    ( -1, 4, 0, 5 ),
    ( -1, 2, 0, 2 ),
    ( -3, 13, 2, 7 )
\}.
\end{align*}
This list contains several candidates with $m=1$ and $n\geq 2R-1$. However, since we have $b=3$, this implies $n\equiv 0\pmod{4}$, which leaves no relevant solution.
\end{proof}

Since we now know that all nearly-perfect covering codes with covering radius $R=3$ must have the same parameters as the codes in Example~\ref{ex:repeat}, we remark that Theorem~\ref{th:repuniq} immediately implies that they must be equivalent to the codes in the example.

Inspection of the proof of Theorem~\ref{th:codesR3} reveals another result:

\begin{corollary}
For any fixed $R\geq 3$, there are at most a finite number of non-perfect $(n,M,d)R$ nearly-perfect covering codes.
\end{corollary}

\begin{proof}
We use the same notation as the proof of Theorem~\ref{th:codesR3}. The first half of the proof of Theorem~\ref{th:codesR3} ends by needing to solve~\eqref{eq:thuemahler1} and~\eqref{eq:thuemahler1R}.

We now contend that for any $1\leq b\leq R$, the polynomial $F_b(n)$ has at least two distinct complex roots. Let us first examine the case $1\leq b\leq R-1$. We observe that $F_b(n)$ is monic of degree $R+1$. The coefficient of $n^R$ is $1-\binom{R+1}{2}$. If $F_b(n)$ had only one distinct root, then we would have $F_b(n)=(n-\alpha)^{R+1}$ for some $\alpha\in\C$. But then the coefficient of $n^R$ would equal $-(R+1)\alpha$, giving us
\[
\alpha=\frac{R}{2}-\frac{1}{R+1}.
\]
For $R\geq 3$ this $\alpha$ is rational, but not an integer. Write $\alpha=\frac{v}{w}$  for $v,w\in\Z$, $\gcd(v,w)=1$. Since the coefficients of $F_b(n)$ are integers, we have that $w$ must divide the coefficient of $n^{R+1}$. But the polynomial is monic, so $w=\pm 1$, and $\alpha$ is an integer, a contradiction.

In a similar manner, we contend that for the case $b=R$, $F_R(n)$ has at least two distinct complex roots. Again, $F_R(n)$ is monic, and the coefficient of $n^{R-1}$ is $-\binom{R}{2}$, while the constant term is $R!$. Like before, if there is a single distinct complex root $\alpha$ to $F_R(n)$, then
\[
\alpha=\frac{R-1}{2}.
\]
If $R$ is even, $\alpha$ is not an integer, and like above, we get a contradiction since the polynomial is monic. If $R$ is odd, the constant term of $(x-\alpha)^R$ is negative, while the constant term of $F_R(n)$ is $R!$ which is positive. In all cases we get a contradiction.

If $S$ is a set of prime numbers, for any non-zero integer $m$, define $[m]_S$ to be the largest integer dividing $m$ all of whose prime factors are in $S$. Let $S=\set{p_1,\dots,p_t}$, the primes from~\eqref{eq:thuemahler1} and~\eqref{eq:thuemahler1R}. By~\cite[Corollary 1.5]{GroVin13} or~\cite[Theorem 2.2]{BugEveGyo18}, there are constants $\kappa_1,\kappa_2>0$, such that
\[
[F_b(n)]_S \leq \kappa_2\abs*{F_b(n)}^{1-\kappa_1}.
\]
By~\eqref{eq:thuemahler1} and~\eqref{eq:thuemahler1R},
\[
[F_b(n)]_S = \abs*{F_b(n)}.
\]
It follows that
\[
\abs*{F_b(n)} \leq \kappa_2^{1/\kappa_1}.
\]
Since $F_b(n)$ is non-constant, this implies a constant upper bound to any $n$ that satisfies~\eqref{eq:thuemahler1} or~\eqref{eq:thuemahler1R}.

Since $R$ is fixed, once $n$ is found, by~\eqref{eq:newdef} there is only one possible value for $M$. Additionally, by Lemma~\ref{lemma:quasiperfect}, there are only two options for $d\in\set{2R-1,2R}$. Finally, given the parameters $(n,M,d)R$, there is a finite number of translations and permutations to produce equivalent codes.
\end{proof}

We note that the Thue-Mahler machinery used in this paper requires polynomials of degree at least $3$. This holds for any covering radius $R\geq 3$. For smaller radii we get lower-degree polynomials, e.g., for $R=2$ the case $b=2$ in the proof of Theorem~\ref{thm:radius2} is only quadratic. The reason the Thue-Mahler machinery may fail at lower degrees is that there may be infinitely many integer solutions. This is indeed the case for $R=1$, where there exist infinite families of nearly-perfect covering codes~\cite{BorEtzRot25}, or of nearly-perfect error-correcting codes with packing radius $e=1,2$.

\section{Towards a Full Characterization}
\label{sec:groupalgebra}

A convenient way to represent binary vectors is by using monomials on the variables $z_1,z_2,\dots,z_n$, where we adopt the convention $z_i^2=1$. The vector $x=(x_1,x_2,\dots, x_n)\in\F_2^n$ is represented by the monomial $z_1^{x_1}\dots z_n^{x_n}$ which we abbreviate $z^x$. Let $G$ be the multiplicative group of all possible $z^v$. The group algebra $\Q G$ over the rational numbers consists of all formal sums $\sum_{v\in\F_2^n}a_v z^v$ with rational coefficients $a_v\in\Q$. For more details, see \cite[Chapter 5]{MacSlo78}. We denote 
\[Y_i\eqdef\sum_{\substack{v\in \F_2^n\\ \wt(v)=i}}z^v,
\]
and given a code $C\subseteq\F_2^n$, by abuse of notation, we define $C\eqdef\sum_{c\in C}z^c$.

The general strategy to prove a full characterization of perfect and nearly-perfect error-correcting codes consists of deriving an equality of the form 
\[
    C\cdot \sum_{i}\beta_i Y_i=\sum_{v\in\F_2^n}z^v
\]
in the group algebra. For perfect codes, the equality in the sphere packing bound readily gives us such an expression:
\[
C\cdot \sum_{i=0}^{e}Y_i = \sum_{v\in \F_2^n}z^v.
\]

For nearly-perfect error-correcting codes with minimum distance $d=2e+1$, the following equality holds~\cite[Theorem 4.4]{GoeSno72}
\[
C\cdot \left(\sum_{i=0}^{e-1}Y_i+\frac{1}{\lfloor (n+1)/(e+1)\rfloor}(Y_e+Y_{e+1})\right)= \sum_{v\in \F_2^n}z^v.
\]
The proof relies on the following regularity result:

\begin{lemma}[{\cite[Lemma 3.1]{GoeSno72}}]
    A nearly-perfect $e$-error-correcting code of length $n$ satisfies: 
    \begin{itemize}
        \item[1)] Any vector at distance greater than $e$ from every codeword is at distance $e+1$ from exactly $\lfloor n/(e+1)\rfloor$ codewords.
        \item[2)] Any vector at distance $e$ from a given codeword is at distance $e+1$ from exactly ${\lfloor (n-e)/(e+1)\rfloor}$ other codewords.
    \end{itemize}
\end{lemma}

A natural question arises: Can this same strategy be applied for the case of nearly-perfect covering codes? Again, by abuse of notation, we let $C_R\eqdef\sum_{v\in C_R}z^v$. Then we have 
\begin{equation}\label{eq:groupalgebra1}
    C\cdot \sum_{i=0}^{R-1} Y_i + C_R = \sum_{v\in \F_2^n}z^v\\
\end{equation}
because the code is quasi-perfect (see Lemma~\ref{lemma:quasiperfect}). In order to derive a second equality, we can study the value of $C\cdot \sum_{i=0}^{R+1}Y_i$. If $d=2R$ is even,
\begin{equation}\label{eq:groupalgebra2}
    C\cdot \sum_{i=0}^{R+1}Y_i = \sum_{v\in\F_2^n}z^v+C_R\cdot\left(a-1\right)+\sum_{v\in C_{R-1}}(|B_{R+1}(v)\cap C|-1)z^v.
\end{equation}
The equation holds because balls of radius $R+1$ centered at codewords can only intersect at vectors in $C_R$ or $C_{R-1}$ if $d$ is even, so vectors at distance $\leq R-2$ are counted once in the left-hand side sum, vectors in $C_R$ are counted $a$ times by Lemma \ref{lemma:epsilonequality}, and a vector in $C_{R-1}$ is counted $|B_{R+1}(v)\cap C|$ times. So we see that in order to get a useful expression, we would need a regularity result showing that $|B_{R+1}(v)\cap C|$ is constant for all $v\in C_{R-1}$. When $R=1$ this has been proven to be true in \cite[Theorem 6]{BorEtzRot25}. In that case, by combining \eqref{eq:groupalgebra1} and \eqref{eq:groupalgebra2} we retrieve the same equation as in \cite[Section IV.B]{BorEtzRot25}:
\[C\cdot (n/2+Y_1+Y_2)=(n/2+1)\sum_{v\in \F_2^n}z^v.\] 

We now show generalizations of the known regularity results for nearly-perfect covering codes in the case $R=1$, as stepping stones for a future group-algebra equality. The following is a generalization of~\cite[Lemma 5]{BorEtzRot25}:

\begin{theorem}\label{thm:lowerboundR+1}
Let $C$ be a non-perfect $(n,M,d)R$ nearly-perfect covering code satisfying one of the following:
\begin{itemize}
    \item $R=1$, or 
    \item $d$ is even and  $(n-R+1)(a-1)\not\equiv 0 \pmod {R+2}$.
\end{itemize}
 Then for any $x\in C_{R-1}$, we have
    \[
    \abs*{B_{R+1}(x)\cap C}\geq 2.
    \]
\end{theorem}
\begin{proof}
By translating the code, we can assume $x=0$. Let $c\in C$ be the unique codeword covering $x=0$ with $\wt(c)=R-1$ (uniqueness follows from the fact that the code is quasi-perfect by Lemma~\ref{lemma:quasiperfect}). Without loss of generality, $\supp(c)=\set{n-R+2,\dots,n}$, and define $J\eqdef \set{1,2,\dots,n-R+1}$. We denote, for all $0\leq i\leq n$,
\begin{align*}
\cS_i &\eqdef \set*{z\in\F_2^n : \wt(z) = i}, &
\cS'_i &\eqdef \set*{z\in\cS_i : \supp(z) \subseteq J}.
\end{align*}
Notice that if $R=1$ then $c=x=0$, $J=\set{1,2,\dots,n}$ and $\cS_i=\cS'_i$ for every $0\leq i\leq n$. Consider any standard unit vector $e_i$, $i\in J$. Note that $d(e_i,c)=R$. If $d(e_i,C)\leq R-1$, then there exists some codeword $c'\in C$ such that $d(e_i,c')\leq R-1$. But then obviously $d(0,c')\leq R$, and the claim is proved since $c,c'\in B_{R+1}(0)$.

Let us therefore continue with the case $e_i\in C_R$ for all $i\in J$. Thus, by Lemma~\ref{lemma:epsilonequality}, $g(B_1(e_i))=b>0$ for all $i\in J$, where the inequality follows from the fact that the code is non-perfect. 

Suppose by contradiction that $\abs{B_{R+1}(0)\cap C}\leq 1$, which would mean that the only codeword at a distance $\leq R+1$ from $0$ would be $c$. This implies that no vector in $\cS_1$ can be over-covered. Thus, the over-covering $g(B_1(e_i))$, $i\in J$, can only come from over-covering of elements in $\cS_2$. Moreover, only vectors in $\cS_2'$ can be over-covered. If $R=1$ this is trivially true since $\cS_2=\cS'_2$. On the other hand, if $R>1$ vectors in $\cS_2-\cS'_2$ are at distance $\leq R-1$ from $c$, so not in $C_R$. Recall from Lemma \ref{lemma:muequalitycond} that if $d$ is even, only vectors in $C_R$ can be over-covered, and thus we would have a contradiction. Since each over-covered vector in $\cS'_2$ will appear in two $B_1(e_i)$, we have:
\[
2 g(\cS_2')=2 g(\cS_2)=\sum_{i\in J}g(B_1(e_i))=b(n-R+1).
\]
    
The over-covering $g(B_1(e_i))$, $i\in J$, must therefore come from codewords in $\cS_{R+2}$. We claim that any codeword in $\cS_{R+2}$ is actually in $\cS'_{R+2}$: suppose otherwise that there existed $c'\in\cS_{R+2}$, with $i'\in \supp (c')\cap \supp(c)$. Then, the vector $z\in\cS_2$ with support $\set{i',i''}$, $i''\in \supp (c')$, $i'\neq i''$, would satisfy $d(z,c')=R$ and $d(z,c)\leq R-1$, which would mean $z$ is over-covered and not in $C_R$, a contradiction. Also, notice that a vector in $\cS_2'$  needs to be covered by a codeword in $\cS_{R+2}$ (because $c$ does not cover it).

All this entails that there are $\binom{n-R+1}{2}$ vectors in $\cS_2'$ which need to be covered by codewords in $\cS_{R+2}'$. Notice that $z\in\cS_2'$ being covered by $c'\in\cS_{R+2}'$ amounts to $\supp (z)\subseteq \supp (c')$. The supports of codewords in $\cS'_{R+2}$ form a $(2,R+2,n-R+1)$ covering design. By the well-known Sch{\"o}nheim bound~\cite[Theorem I]{Sch64}, the number of these codewords must be lower bounded by
\begin{align*}
\abs*{\cS'_{R+2}\cap C} &\geq \ceil*{\frac{n-R+1}{R+2}\ceil*{\frac{n-R}{R+1}}} = \ceil*{\frac{(n-R+1)(a-1)}{R+2}} \\
& \geq \frac{(n-R+1)(a-1)+1}{R+2},
\end{align*}
where the last inequality follows from our requirement that $(n-R+1)(a-1)\not\equiv 0 \pmod{R+2}$ when $d$ is even, and in the case $R=1$, it is known that $n=2^m$ for some $m\geq 1$  \cite{BorEtzRot25}, so $a-1= \lceil (2^m+1)/2 \rceil -1=2^{m-1}$ and thus $3=R+2\nmid 2^{2m-1}=(n-R+1)(a-1)$.

We now have,
    \begin{multline*}
        \frac{(n-R+1)(a-1)+1}{R+2}\binom{R+2}{2}\leq \abs*{\cS_{R+2}'\cap C}\binom{R+2}{2}=f(\cS_2')\\
        =\abs*{\cS_2'}+g(\cS_2') = \binom{n-R+1}{2}+\frac{b(n-R+1)}{2}. 
    \end{multline*}
Comparing the two endpoints of this chain, and after some algebraic manipulation, we arrive at $R\leq -1$, a contradiction.
\end{proof}

\begin{remark}\label{remark:regularityEven}
If $d$ is odd and $R>1$ then Theorem \ref{thm:lowerboundR+1} is false in general: for the code $C'$ from Example~\ref{ex:repeat} we have $a=2$ and $b=1$, so $(n-R+1)(a-1)=m+1\not\equiv 0\pmod {m+2}$. However, the vector $x=0^{m+1}1^{m-1}$ is at distance $m-1=R-1$ from the code, and $B_{R+1}(x)\cap C = \set{0^{2m}}$ whenever $m>1$.
\end{remark}

The following Corollary generalizes \cite[Theorem 6]{BorEtzRot25}:

\begin{corollary}\label{coro:alpha}
    Let $C$ be a non-perfect $(n,M,d)R$ nearly-perfect covering code satisfying the same condition as in Theorem~\ref{thm:lowerboundR+1}. Define
    \[
    \alpha\eqdef \frac{1}{\abs{C_{R-1}}}\sum_{x\in C_{R-1}}\abs*{B_{R+1}(x)\cap C}.
    \]
    Then 
    \[2\leq \alpha \leq 3-\frac{2}{R+1}.\]
    In particular, if $R=1$, then $\alpha=2$ and $\abs{B_{R+1}(x)\cap C}=2$ for every $x\in C_{R-1}$. Additionally, if $R>1$ then $\alpha=2$ under the assumption that $\abs{B_{R+1}(x)\cap C}=\alpha$ is a constant that does not depend on $x\in C_{R-1}$.
\end{corollary}

\begin{proof}
We observe that $\alpha$ is the average of $\abs{B_{R+1}(x)\cap C}$ as $x$ ranges over $C_{R-1}$. Thus, by Theorem~\ref{thm:lowerboundR+1}, $\alpha\geq 2$. Denote the volume of a ball of radius $r$ by $V_r\eqdef \abs{B_r(0)}=\sum_{i=0}^r\binom{n}{i}$. Then,
\begin{align}
        M V_{R+1}&=\sum_{x\in\F_2^n}\abs*{B_{R+1}(x)\cap C}=\parenv*{\sum_{x\in C_R}+\sum_{x\in C_{R-1}}+\sum_{x :d(x,C)\leq R-2}}\abs*{B_{R+1}(x)\cap C}\label{eq:coro1}\\
        &= a\abs*{C_R}+\alpha \abs*{C_{R-1}}+M V_{R-2}\label{eq:coro2}\\
        &=a(2^n-M V_{R-1})+\alpha M \binom{n}{R-1}+M V_{R-2}\label{eq:coro3}\\
        &=a\sparenv*{M\parenv*{V_R-\binom{n}{R}\frac{b}{n+1+b-2R}}-M V_{R-1}}+\alpha M \binom{n}{R-1}+M V_{R-2}\label{eq:coro4},
\end{align}
where the third sum in \eqref{eq:coro1} is $0$ if $R=1$, \eqref{eq:coro2} follows from Lemma \ref{lemma:epsilonequality}, the definition of $\alpha$, and Lemma \ref{lemma:muequalitycond} when $d$ is even, \eqref{eq:coro3} follows from Lemma \ref{lemma:quasiperfect}, and \eqref{eq:coro4} from Corollary \ref{coro:defwith2Rreplaced}.

Dividing by $M$, we obtain,
    \[
    V_{R+1}-V_{R-2}-\alpha \binom{n}{R-1}-a(V_R-V_{R-1})+\frac{ab}{n+1+b-2R}\binom{n}{R}=0.
    \]
Using the expressions for the volume, this gives,
    \[
    \binom{n}{R+1}+\binom{n}{R}+\binom{n}{R-1}-\alpha \binom{n}{R-1}-a\binom{n}{R}+\frac{ab}{n+1+b-2R}\binom{n}{R}=0.
    \]
After dividing by $\binom{n}{R}$,
    \[
    \frac{n-R}{R+1}+1-(\alpha-1)\frac{R}{n-R+1}-a+\frac{ab}{n+1+b-2R}=0.
    \]
    Multiplying by $R+1$, and recalling $a(R+1)=n+1+b$, we obtain,
    \[
    \frac{b(n+1+b)}{n+1+b-2R}-b-(\alpha-1)\frac{R(R+1)}{n-R+1}=0.
    \]
%    Multiplying by $n-R+1$, and rearranging gives
%    \[
%    \frac{b(n+1+b)(n-R+1)}{n+1+b-2R}=b(n-R+1)+(\alpha-1)R(R+1)
%    \]
%    \[
%    b(n-R+1)((n+1+b)-(n+1+b-2R))=(\alpha-1)(n+1+b-2R)R(R+1)
%    \]
%    \[
%    2b(n-R+1)=(\alpha-1)(R+1)(n+1+b-2R)
%    \]
%    \[
%    b(2(n-R+1)-(\alpha-1)(R+1))=(\alpha-1)(R+1)(n+1-2R),
%    \]
%    from which we deduce
After rearranging we get,
    \[
    b\cdot (2(n-R+1)-(\alpha-1)(R+1))= (R+1) (\alpha-1)(n+1-2R).
    \]
If $R=1$ or $d$ is even, the right hand side is strictly positive, and thus so is $2(n-R+1)-(\alpha-1)(R+1)$. One can verify now that if $\alpha >3-\frac{2}{R+1}$, then
\[
\frac{b}{R+1}=\frac{(\alpha-1)(n+1-2R)}{2(n-R+1)-(\alpha-1)(R+1)} > \frac{R}{R+1}.
\]
But then $b> R$, which is a contradiction. We therefore proved that $2\leq \alpha \leq 3-\frac{2}{R+1}$.

For the last claim, if $\abs{B_{R+1}(x)\cap C}=\alpha$ is constant and does not depend on $x$, then $\alpha$ must be an integer. The only possible one is $\alpha=2$.
\end{proof}

Unfortunately, Corollary \ref{coro:alpha} is not strong enough to derive a useful group-algebra equality for nearly-perfect covering codes in the case $R>1$. It is quite striking that obtaining such an equality seems to be much more challenging than for perfect and nearly-perfect error-correcting codes, and we do not know if it is indeed possible to derive it. We leave this as an open path for future research.

Finally, we note that finding a group-algebra identity is not only useful to proving the nonexistence of certain codes, but also to studying weight distributions. This technique is employed in~\cite{GoeSno72} for nearly-perfect error-correcting codes, and in~\cite{BorEtzRot25} for nearly-perfect covering codes of radius $R=1$.

\section{Conclusion}
\label{sec:conc}

In this paper we studied nearly-perfect covering codes. We started by adjusting the Van Wee bound to account for the minimum distance of the code, thus opening the door to defining nearly-perfect covering codes with covering radius $R\geq 2$.

We completely classified all nearly-perfect covering codes for $R=2$, showing they must be equivalent to the codes in Example~\ref{ex:repeat} and Example~\ref{ex:len6}. We then extended the treatment to $R=3$, and showed how it may be analyzed for any fixed $R$ greater than this. As a consequence, we proved that for any given $R\geq 3$, there are at most a finite number of nearly-perfect covering codes.

Several open questions remain. First, we do not, at the moment, know how to completely classify all nearly-perfect covering codes. While we do have machinery in place, given in Section~\ref{sec:R3ormore}, it requires the algorithmic solution of Thue-Mahler equations on a case-by-case basis. The group-algebra approach, which was successful in the case of perfect and nearly-perfect error-correcting code classification, seems to be difficult as we lack the required regularity of coverage.

We curiously note that the codes in Example~\ref{ex:repeat} and Example~\ref{ex:len6}, are all linear codes. Whether all nearly-perfect covering codes have to be equivalent to linear codes, is an interesting question.

Finally, apart from extension to higher radii, there is the obvious extension to larger alphabets. The parameters of perfect codes have been completely determined for larger alphabets~\cite{Tie73}. Even those of nearly-perfect error-correcting codes have been completely analyzed in~\cite{Lin77}. Larger alphabets, however, complicate the analysis following the Van Wee bound of Theorem~\ref{th:vanwee}. For example, if both the alphabet size $q$ and $n$ are even, then the Van Wee bound with $R=1$ is \cite[Corollary 7]{Van91}:
\[
M\geq \frac{q^n}{n(q-1)},
\]
which is not an integer for $q>2$. Thus, we can no longer assume exact equality through the chain of inequalities in the proof of the bound, as we did in the binary case. We expect finding a solution to these open questions to be quite challenging.

\backmatter

\bmhead{Acknowledgments}

This research was supported in part by NSERC under grant no.~RGPIN-2026-06805, and grant no.~CGRS D-611819-2026.
    
\bibliography{allbib}

\end{document}